\documentclass[a4paper,USenglish,cleveref, autoref, thm-restate]{lipics-v2021}
\pdfoutput=1 %uncomment to ensure pdflatex processing (mandatatory e.g. to submit to arXiv)
\hideLIPIcs  

\title{B-Treaps Revised: Write Efficient Randomized Block Search Trees with High Load} 
 
\author{Roodabeh Safavi}{Institute of Science and Technology Austria (ISTA), Am Campus 1, A-3400 Klosterneuburg, Austria}{roodabehsafavi@gmail.com}{https://orcid.org/0000-0003-4516-4212}{}
\author{Martin P. Seybold}{University of Vienna, Faculty of Computer Science, Theory and Applications of Algorithms, Währinger Straße 29, A-1090 Vienna, Austria}{mpseybold@gmail.com}{https://orcid.org/0000-0001-6901-3035}{}

\authorrunning{Roodabeh Safavi and Martin P. Seybold} 
\Copyright{Roodabeh Safavi and Martin P. Seybold}


\relatedversion{} %arxiv
\funding{This work was supported under the Australian Research Council Discovery Projects funding scheme (project number DP180102870).
This work was further supported by the European Research Council under the European Union's Horizon 2020 research and innovation funding scheme (grant number 101019564 ``The Design of Modern Fully Dynamic Data Structures'') \includegraphics[height=11pt]{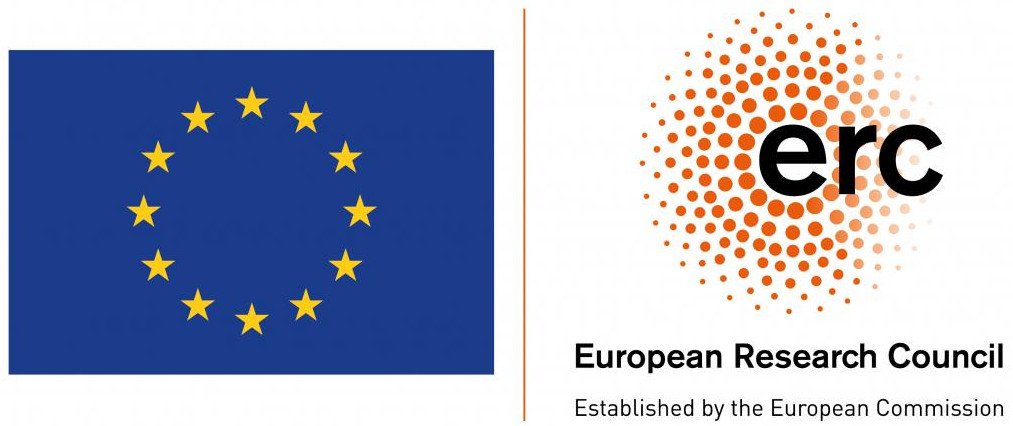}.}
\nolinenumbers %uncomment to disable line numbering

\EventEditors{John Q. Open and Joan R. Access}
\EventNoEds{2}
\EventLongTitle{42nd Conference on Very Important Topics (CVIT 2016)}
\EventShortTitle{CVIT 2016}
\EventAcronym{CVIT}
\EventYear{2016}
\EventDate{December 24--27, 2016}
\EventLocation{Little Whinging, United Kingdom}
\EventLogo{}
\SeriesVolume{42}
\ArticleNo{23}
\usepackage{multirow}

\usepackage{hyperref, cancel}
\usepackage{algpseudocode, algorithm}
\usepackage{graphicx, placeins, breqn}
\usepackage{makecell}
\usepackage{amsmath}
\allowdisplaybreaks

\usepackage{xcolor}
\usepackage{hyperref}
\hypersetup{
	colorlinks,
	linkcolor={green!75!black},%{red!90!black},
	citecolor={blue!75!black},%{green!80!black},
	urlcolor={blue!85!black}
}

\algnewcommand\algorithmicforeach{\textbf{for each}}
\algdef{S}[FOR]{ForEach}[1]{\algorithmicforeach\ #1\ \algorithmicdo}

\newcommand{\E}{\mathbb{E}}
\newcommand{\eps}{\varepsilon}
\renewcommand{\O}{\mathcal{O}}
\newcommand{\ceil}[1]{{\lceil #1 \rceil}}
\newcommand{\floor}[1]{{\lfloor #1 \rfloor}}
\newcommand{\Perm}{\operatorname{Perm}}

\begin{document}

\maketitle

\begin{abstract}
Uniquely represented data structures represent each logical state with a unique storage state.
We study the problem of maintaining a dynamic set of $n$ keys from a totally ordered universe in this context.

We introduce a two-layer data structure called $(\alpha,\varepsilon)$-Randomized Block Search Tree (RBST) that is uniquely represented and suitable for external memory.
Though RBSTs naturally generalize the well-known binary Treaps, several new ideas are needed to analyze the {\em expected} search, update, and storage, efficiency in terms of block-reads, block-writes, and blocks stored.
We prove that searches have $O(\varepsilon^{-1} + \log_\alpha n)$ block-reads, that $(\alpha, \varepsilon)$-RBSTs have an asymptotic load-factor of at least $(1-\varepsilon)$ for every $\varepsilon \in (0,1/2]$, and that dynamic updates perform $O(\varepsilon^{-1} + \log_\alpha(n)/\alpha)$ block-writes, i.e. $O(1/\varepsilon)$ writes if $\alpha=\Omega(\frac{\log n}{\log \log n})$.
Thus $(\alpha, \varepsilon)$-RBSTs provide improved search, storage-, and write-efficiency bounds in regard to the known, uniquely represented B-Treap [Golovin; ICALP'09].
\end{abstract}

\section{Introduction}
\label{sec:intro}
Organizing a set of keys %, from a totally ordered universe, 
such that insertions, deletions, and searches, are supported is one of the most basic problems in computer science that drives the development and analysis of search structures with various guarantees and performance trade-offs.
Binary search trees for example have several known height balancing methods (e.g. AVL and Red-Black trees) and weight balancing methods (e.g. BB[a] trees), some with $\O(1)$ writes per update (e.g.~\cite{ElmasryKAH19}). 

Block Search Trees are generalizations of binary search trees that store in every tree node up to $\alpha$ keys and $\alpha+1$ child pointer, where the array size $\alpha$ of the blocks is a (typically) fixed parameter.
Since such layouts enforce a certain data locality, block structures are central for read and write efficient access in machine models with a memory hierarchy.

\subparagraph{External Memory and deterministic B-Trees variants}
In the classic External Memory~(EM) model, the $n$ data items of a problem instance must be transferred in blocks of a fixed size $B$ between a block-addressible external memory and an internal Main Memory~(MM) that can store up to $M$ data items and perform computations.
Typically, $n \gg M > B$ and the MM can only store a small number of blocks, say $M/B=\O(1)$, throughout processing.
Though Vitter's Parallel Disk Model can formalize more general settings (see~\cite[Section~$2.1$]{Vitter01}), we are interested in the most basic case with one, merely block-addressible, EM device and one computing device with MM.
The cost of a computation is typically measured in the number of block-reads from EM~(Is), the number of block-writes to EM~(Os), or the total number of IOs performed on EM.
For basic algorithmic tasks like the range search, using block search trees with $\alpha=\Theta(B)$ allows obtaining algorithms with asymptotic optimal total IOs (see, e.g., Section~$10$ in the survey~\cite{Vitter01}).

Classic B-Trees~\cite{BayerM72} guarantee that every block, beside the root, has a load-factor of at least $1/2$ and that all leaves are at equal depth, using a deterministic balance criteria.
The B*-Tree variant~\cite[Section~6.2.4]{KnuthVol3} guarantees a load-factor of at least $2/3$ based on a (deterministic) overflow strategy that shares keys among the two neighboring nodes on equal levels.
Yao's Fringe Analysis~\cite{Yao78} showed that inserting keys under a random order in (deterministic balanced) B-Trees yields an expected asymptotic load-factor of $\ln(2)\approx 69\%$, and the expected asymptotic load-factor is $2\ln(3/2)\approx81\%$ for B*-Trees.
For general workloads however, maintaining higher load-factors (with even wider overflow strategies) further increases the write-bounds of updates in the block search trees~\cite{Baeza-Yates89a,Kuspert83}.

The popular, e.g.~\cite{EsmetBFK12,JannenYZAEJMPRW15}, write-optimized B$^\eps$-Tree~\cite{BrodalF03,BenderFFFKN07} provides smooth trade-offs between 
$ \O( \log_{1+\alpha^\eps} (n) / \alpha^{1-\eps})$ 
% $ \O( \frac{1}{\eps}\frac{\log_\alpha n}{\alpha^{1-\eps}})$ 
amortized block writes per insertion and 
$\displaystyle \O(\log_{1+\alpha^\eps} n)$ 
% $\O( \frac{1}{\eps} \log_\alpha n)$ 
block reads for searches.
E.g. tuning parameter $\eps=1$ retains the bounds of B-Trees and $\eps=1/2$ provides improved bounds. 
In the fully write-optimized case ($\eps=0$) however, searches have merely a $\O(\log n)$ bound for the number of block-reads.
The main design idea to achieve this is to augment the non-leave nodes of a B-Tree with additional `message buffers' so that a key insertion can be stored close to the root. Messages then only need propagation, further down the tree, once a message buffer is full (cf.~\cite[Section~$3.4$]{BrodalF03} and~\cite{BenderFFFKN07}). 

Clearly, the state of either of those (deterministic) structures depends heavily on the actual sequence in which the keys were inserted in them.

\subparagraph{Uniquely Represented Data Structures}
For security or privacy reasons, some applications, see e.g.~\cite{GolovinPhd}, require data structures that do not reveal any information about historical states to an observer of the memory representation, i.e. evidence of keys that were deleted or evidence of the keys' insertion sequence.
Data structures are called \emph{uniquely represented} (UR) if they have this strong history independence property.
For example, sorted arrays are UR, but the aforementioned search trees are not UR due to, e.g., deterministic rebalancing.
Early results~\cite{Snyder77, SundarT94, AnderssonO95} show lower and upper bounds on comparison based search in UR data structures on the pointer machine.
Using UR hash tables~\cite{BlellochG07, NaorT01} in the RAM model, also pointer structures can be mapped uniquely into memory.
(See Theorem~$4.1$ in~\cite{BlellochG07} and Section~$2$ in~\cite{BlellochGV08}.) 

The Randomized Search Trees are defined by inserting keys, one at a time, in order of a random permutation into an initially empty \emph{binary} search tree.
The well-known Treap~\cite{SeidelA96} maintains the tree shape, of the permutation order, and supports efficient insertions, deletions, and searches (see also~\cite[Chapter~1.3]{mulmuley}). 
Any insertion, or deletion performs expected $\O(1)$ rotations (writes) and searches read expected at most $2 \ln (n+1)$ tree nodes, which is particularly noteworthy for machine models with asymmetric cost or concurrent access.

Golovin~\cite{Golovin09} introduced the B-Treap as first UR data structure for Byte-addressable External Memory~(BEM).
B-Treaps are due to a certain block storage layout of an associated (binary) Treap.
I.e. the block tree is obtained, bottom-up, by iteratively placing sub-trees, of a certain size, in one block node\footnote{E.g. byte-addressable memory is required for intra-block child pointers and each key maintains a size field that stores its sub-tree size (within the associated binary Treap).}.
The author shows bounds of B-Treaps, under certain parameter conditions:
If $\alpha = \Omega \left( (\ln n)^{\frac{1}{1-\varepsilon} } \right)$, for some $\varepsilon>0$, then updates have expected $\O(\frac{1}{\varepsilon} \log_\alpha n)$ block-writes and range-queries have expected $\O(\frac{1}{\varepsilon} \log_\alpha n + k/\alpha)$ block-reads, where $k$ is the output size.
If $\alpha = \Omega \left( n^{\frac{1}{2}-\varepsilon} \right)$, for some $\varepsilon>0$, then depth is with high probability $\O(\frac{1}{\varepsilon} \log_\alpha n)$ and storage space is linear (see Theorem~1 in \cite{Golovin09}).
The paper also discusses experimental data, from a non-dynamic implementation, that shows an expected load-factor close to $1/3$ (see Section~6 in \cite{Golovin09}).
Though the storage layout approach allows leveraging known bounds of Treaps to analyze B-Treaps, the block size conditions, e.g. $\alpha=\Omega( \ln^2 n )$, are limiting for applications of B-Treaps in algorithms.

\begin{table} \renewcommand{\arraystretch}{1.2}% for the vertical padding
\begin{tabular*}{\textwidth}{ccc|rll} 
~ & Data Structure & Model & ~ & \#~Blocks & ~                    \\  \hline
\multirow{6}{*}{non-UR} & \multirow{3}{*}{B-Tree~\cite{BayerM72}}  & 
\multirow{3}{*}{EM}
    & Size 
    & $\O(n/\alpha)$  
    &  \\  
    & & & Updates  & \multirow{2}{*}{$\O( \log_\alpha n)$} & \multirow{2}{*}{} \\
    & & & Search & & 
\\ \cline{2-6}
        & \multirow{3}{*}{B$^\eps$-Tree~\cite{BrodalF03}}  & 
\multirow{3}{*}{EM}
    & Size 
    & $\O(n/\alpha)$  
    & \\  
    & & & Updates  & $\O\left( \frac{1}{\eps} \frac{\log_\alpha n}{\alpha^{1-\eps}} \right)$ & \multirow{2}{*}{ } \\
    & & & Search & $\O\left(\frac{1}{\eps} \log_{\alpha}n\right)$& 
\\ \hline %\hline
\multirow{7}{*}{UR} & \multirow{3}{*}{B-Treap~\cite{Golovin09}}  & 
\multirow{3}{*}{BEM}
    & Size 
    & $\O(n/\alpha)$  
    & 
    %    $\alpha=\O\left(n^{\frac{1}{2}-\varepsilon}     \right)$.
    \\  
    & & & Updates write & \multirow{2}{*}{$\O\left(\frac{1}{\eps} \log_\alpha n \right)$} & \multirow{2}{*}{, if $\alpha=\Omega\left( \ln^{\frac{1}{1-\varepsilon}}n \right)$.}\\
    & & & Search reads & & 

\\ \cline{2-6}
% \multirow{4}{*}{UR}
& \multirow{5}{*}{$(\alpha,\eps)$-RBST} & \multirow{4}{*}{EM} &  \multirow{2}{*}{Size}
& $\O(n/\alpha)$ & \\ 
&&&& $\leq (1+\eps)n/\alpha$  &, if $n=\omega(\alpha^2/\eps)$.\\
\cline{4-5}
&&& \multirow{2}{*}{Updates write} & $\O\left(\frac{1}{\eps} + \frac{\log_\alpha n}{\alpha}\right)$ & \\
&&&&$\O(1/\eps)$ &, if $\alpha=\Omega\left( \frac{\log n}{\log \log n} \right)$. \\
\cline{4-5}
&&& Search reads & $\O(\frac 1 \eps + \log_\alpha n)$ &  ~ \\

\end{tabular*}
\caption{Overview of known non-UR and UR data structures for EM and the proposed RBSTs.} \label{tabel:overview} \vspace{-.25cm}
\end{table}

\subsection{Contribution and Paper Outline}
We propose a two-layer randomized data structure called $(\alpha, \eps)$-Randomized Block Search Trees (RBSTs), where $\alpha$ is the block size and $\eps \in (0,1/2]$ a tuning parameter for the space-efficiency of our secondary layer (see Section~\ref{sec:RBSTs}).
Without the secondary structures, the RBSTs are precisely the distribution of trees that are obtained by inserting the keys, one at a time, under a random order  into an initially empty block search tree (e.g. $\alpha = 1$ yields the Treap distribution). 
RBSTs are UR and the algorithms for searching are simple.
We give a partial-rebuild algorithm for dynamic updates that occupies $\O(1)$ blocks of temporary MM storage and performs a number of write operation to EM that is proportional to the structural change of the RBST (see Section~\ref{sec:algo-updates}).

We prove three performance bounds for $(\alpha,\eps)$-RBSTs in terms of block-reads, block-writes, and blocks stored.
Section~\ref{sec:search-and-depth} shows that searches read expected $\O(\eps^{-1}+ \log_\alpha n)$ blocks.
Section~\ref{sec:size} shows that $(\alpha, \eps)$-RBSTs have an expected asymptotic load-factor of at least $(1 -\eps)$ for every $\eps \in (0,1/2]$.
Combining both analysis techniques, we show in Section~\ref{sec:efficient-updates} that dynamic updates perform expected $\O(\eps^{-1}+\log_\alpha(n)/\alpha)$ block-writes.

Thus, RBSTs are simple UR search trees, for all block sizes $\alpha \geq 1$, that simultaneously provide improved search, storage utilization, and write-efficiency, compared to the bottom-up B-Treap~\cite{Golovin09}.
See Table~\ref{tabel:overview} for a comparison with known UR and non-UR data structures for EM. 
To the best of our knowledge, $(\alpha,\eps)$-RBSTs are the first structure that provides the optimal search bound while being fully-write efficient and storage-efficient.
Central for the design of our block search tree is the secondary layer that uses a fan-out for the block nodes that is proportional the subtree's weight, i.e. the number of keys that are stored in it.

% \newpage
\section{Randomized Block Search Trees} 
\label{sec:RBSTs}
Unlike the well-known B-Trees that store between $\alpha$ and $\alpha/2$ keys in every block and have all leaves at equal depth, our definition of $(\alpha,\eps)$-RBSTs does not strictly enforce a minimum load for every single node.
The shape of an RBST $T_X(\pi)$, over a set $X$ of $n$ keys, is defined by an incremental process that inserts the keys $x \in X$ in a random order $\pi$, in an initially empty block search tree structure, i.e. with ascending priority values $\pi(x)$.
The actual algorithms for dynamic updates are discussed in Section~\ref{sec:algo-updates}.
Next we describe the basic tree structure that supports both, reporting queries (i.e. range search) and (order-)statistic queries (i.e. range counting and searching the $k$-th smallest value).
In Section~\ref{sec:algo-updates} we also discuss how to omit explicitly storing subtree weights, which allows improved update bounds for the case that statistic queries are not required in applications.
Next we define the structure of our block search tree and state its invariants.

Let array size $\alpha \geq 1$ and buffer threshold $\beta \geq 0$ be fixed integer parameters.
Every RBST block node stores a fixed-size array for $\alpha$ keys and $\alpha+1$ child pointers, a parent pointer, and its distance from the root.
As label for the UR of a block-node, we use the key-value $x$ of minimum priority $\pi(x)$ from those keys contained in the node's array.
Every child pointer is organized as a pair, storing both pointer to the child block and the total number of keys $s\geq 1$ contained in that subtree.
(E.g. the weight of a node's subtree is immediately known to the parent, without reading the respective child.)
All internal blocks of an RBST are full, i.e. they store exactly $\alpha$ keys.
Any block has one of two possible states, either \emph{non-buffering}~($s\geq \beta$) or \emph{buffering}~($s<\beta$).
Though $(\alpha,\eps)$-RBSTs will use $\beta=\Theta(\alpha^2/\eps)$,
we first give the definition of the primary tree without the use of buffers ($\beta=0$), i.e. every block remains non-buffering.

\begin{figure} \centering
    \includegraphics[width=\textwidth]{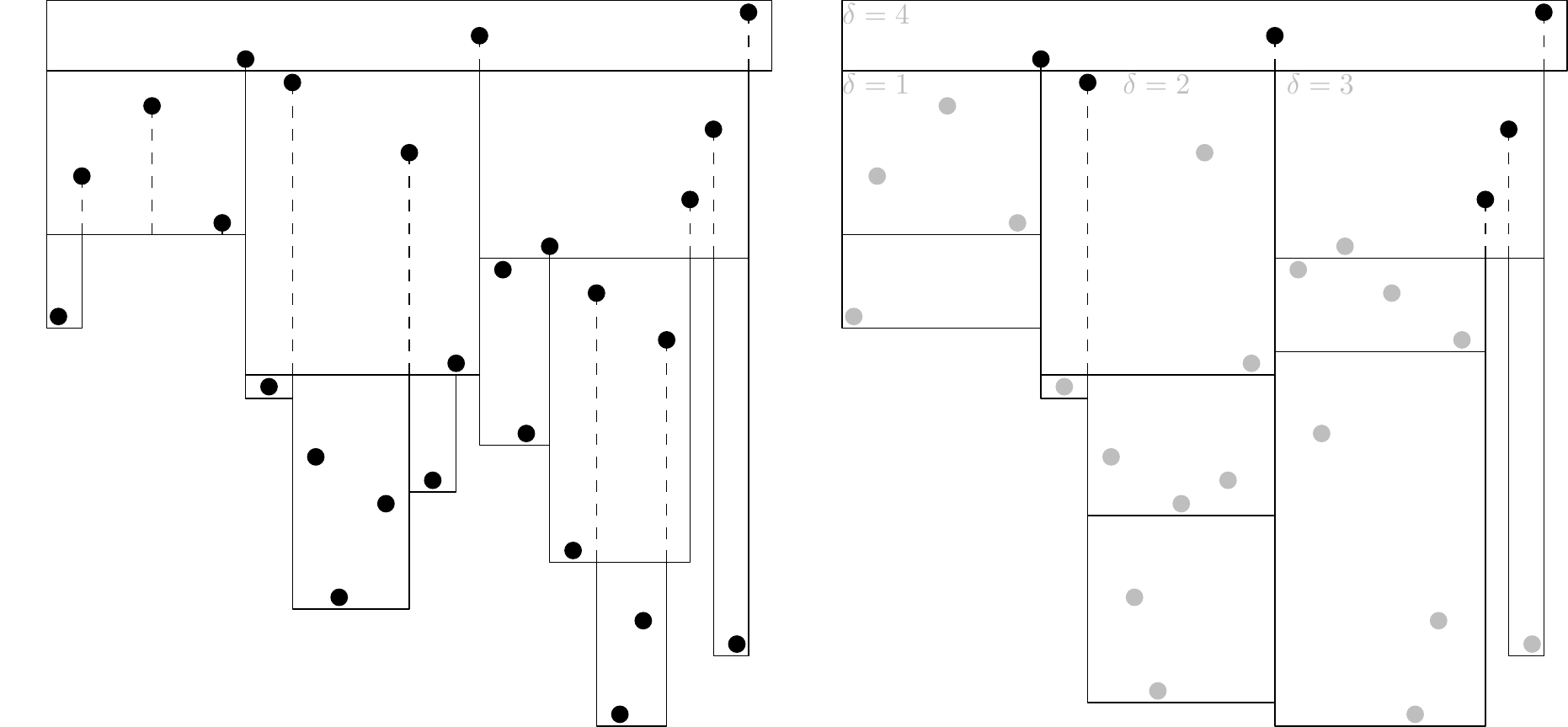}
    \caption{Cartesian representations of two RBSTs that shows the $i$-th largest key $x_i \in X$ with priority value $\pi(i)$ as point $(i,-\pi(i))$ in the plane ($n=26$ and $\alpha=3$).
    The active separators of each block $v$ are shown as vertical lines and the horizontal lines $p^\dagger(v)$.
    The left tree has $\beta=0$ and the right tree has buffering subtrees.
    Note that all, except the last block, on a root-to-leaf path are full.}
    \label{fig:construction}
\end{figure}

For $\beta=0$, the trees are defined by the following incremental process that inserts the keys in ascending order of their priority $\pi(\cdot)$.
Starting at the root block, the insertion of key $x$ first locates the leaf block of the tree using the search tree property.
If the leaf is non-full, $x$ is emplaced in the sorted key array of the node.
Otherwise, a new child block is allocated (according to the search tree property) and $x$ is stored therein.
Thus, any internal block stores those $\alpha$ keys that have the smallest priority values in its subtree\footnote{
For example, $\alpha=1$ and $\beta=0$ yields the well-known Treap and $\alpha>1$ search trees with fan-out $\alpha+1$.
}.
See Figure~\ref{fig:construction} (left) for an example on $26$ keys with $\beta=0$ that consists $12$ blocks, which demonstrates that leaves may merely store one key in their block.
Our tree structure addresses this issue as follows.

For $\beta>0$, the main idea is to delay generating new leaves until the subtree contains a sufficiently large number of keys.
To this end, we use secondary structures that we call \emph{buffers}, that replace the storage layout of all those subtrees that contain $\leq\alpha+\beta$ keys. 
Thus, all remaining block nodes of the primary structure are full and their children are either a primary block or a buffer.
There are several possible ways for UR organization of buffers, for example using a list of at most $\ceil{(\alpha+\beta)/\alpha} = \O(\alpha/\eps)$ blocks that are sorted by key values.
We propose the following UR structure that will result in stronger bounds.

\subparagraph*{Secondary UR Trees: Buffers with weight proportional fan-out}
Our buffers are search trees that consists of nodes similar to the primary tree which also store the keys with the $\alpha$ smallest priority values, from the subtree, in their block. 
However, the fan-out of internal nodes varies based on the number of keys $n\leq \alpha+\beta$ in the subtree.
Define $\delta(n)=\min\{\alpha+1,\max\{1,\ceil{  \frac{n-\alpha}{\rho}  }\}\}$ as the fan-out bound for a subtree of weight $n\leq \alpha+\beta=:\alpha + (\alpha+1)\rho$. 
We defer the calibration of the parameter $\rho=\Theta(\alpha/\eps)$ to Section~\ref{sec:size}.
To obtain UR, the buffer layout is again defined recursively solely based on the set of keys in the buffer and the random permutation $\pi$ that maps them to priority values.

For $\delta=1$, the subtree root has at most one child, which yields a list of blocks.
The keys inside each block are stored in ascending key order, and the blocks in the list have ascending priority values.
That is the first block contains those keys with the $\alpha$ smallest priorities, the second block contains, from the remaining $n-\alpha$ keys, those  with the $\alpha$ smallest priorities, and so forth.

For $\delta \in [2,\alpha+1]$, we also store the keys with the $\alpha$ smallest priorities in the root. 
We call those keys with the $\delta-1$ smallest priority values the \emph{active separators} of this block.
The remaining $n-\alpha$ keys in the subtree are partitioned in $\delta$ sets of key values, using the active separators only.
Each of these key sets is organized recursively, in a buffer, and the resulting trees are referenced as the children of the active separators.

The right part of Figure~\ref{fig:construction} gives an example of the $(\alpha, \eps)$-RBST on the same set of $26$ points that consists of $11$ blocks. 
Some remarks on the definition of primary and secondary RBST nodes are in order.
The buffer is UR since the storage layout only depends on the priority order $\pi$ of the keys, the number of keys in the buffer, and the order of the key values in it.
Note that summing the state values of the child pointers yields $s_1+\ldots+s_\delta + \alpha = n$ the weight of the subtree, without additional block reads.
Moreover, whenever an update in the secondary structure brings the buffer's root above the threshold, i.e.  $s_1+\ldots+s_{\alpha+1} > \beta$, we have that its root contains those keys with the $\alpha$ smallest priorities form the set and all keys are active separators, i.e. the buffer root immediately provides both invariants of blocks from the primary tree. 

Note that this definition of RBSTs not only naturally contains the randomized search trees of block size $\alpha$, but also allows a smooth trade-off towards a simple list of blocks (respectively $\eps=\infty$ and $\eps=0$).
Though leaves are non-empty, it is possible that there exist leaves with load as low as $1/\alpha$ in RBSTs.
Despite its procedural nature, above's definition of the block search tree does not yield a space and I/O efficient algorithm to actually construct and maintain RBSTs.
The reminder of this section specifies the algorithms for searching, dynamic insertion, and dynamic deletion, which we use for computing RBSTs.
Our analysis is presented in the Sections~\ref{sec:search-and-depth},~\ref{sec:size}, and~\ref{sec:efficient-updates}.

\paragraph*{Successor and Range Search in $(\alpha,\eps)$-RBSTs}
As with ordinary B-Trees, we determine with binary search on the array the successor of search value $q$, i.e. the smallest active separator key in the block that has a value of at least $q$, and follow the respective child pointer. 
If the current search node is buffering, i.e. has fewer than $\alpha$ active separators, we also check the set of non-active keys in the block during the search in a local MM variable to the find the successor of $q$.
The search descends until the list, i.e. block nodes with fan-out $\leq 1$, in the respective RBST buffer is found.
Since the lists are sorted by priority, we check all keys, by scanning the $\O(1/\eps)$ blocks of the list, to determine the successor of $q$.
This point-search extends to query-ranges in the ordinary way by considering the two search paths of the range's boundary $[q,q']$.

To summarize, successor and range search occupies $\O(1)$ blocks in main memory, does not write blocks, and the number of block-reads is $\O(D)$, where $D$ is the depth of the RBST.

% \newpage
\subsection{Insertion and Deletion via Partial-Rebuilds}
\label{sec:algo-updates}

As with the Treaps, the tree shape of $(\alpha,\eps)$-RBSTs is solely determined by the permutation $\pi$ of the keys.
For $X \cup \{x\}=X'$, any update method that transforms $T_X(\pi)$ into $T_{X'}(\pi)$ and vice-versa are suitable algorithms for insertions and deletions.
Unlike Treaps, that use rotations to always perform insertions and deletions at the leaf level, rotations in block search trees seem challenging to realize and analyze. 
We propose an update algorithm that rebuilds subtrees in a top-down fashion, which has some vague similarities to the re-balancing mechanism of Scapegoat (Binary Search) Trees.

A naïve update algorithm may seek a full-rebuild of the entire subtree of the block whose array is subject to a key insertion or deletion.
This also allows to maintain the nodes' distance from the root (for certifying search performance).
The main difficulty for update algorithms is to achieve I/O-bounds for EM-access that are (near) proportional to the structural change in the trees, \emph{while} having an $\O(1)$ worst-case bound on the temporary blocks needed in main memory.
The \emph{partial}-rebuild algorithm that we introduce in this section has bounds on the number of block-reads (Is) and block-writes (Os) in terms of the structural change of the update, the expected structural change of RBST updates is analyzed in Section~\ref{sec:efficient-updates}.

\begin{observation} \label{obs:algo}
Let $\pi$ be a permutation on the key universe $U$ and $X' = X \cup \{x\} \subseteq U$ the keys in the tree after and before the insertion of $x$.
Let $m$ be the number of blocks in RBST $T_X(\pi)$ that are different from the blocks in RBST $T_{X'}(\pi)$ and $m'$ the number of blocks in $T_{X'}(\pi)$ that are different from the blocks in $T_{X}(\pi)$.
Let $D$ and $D'$ be the height of the subtrees in $T_X(\pi)$ and $T_{X'}(\pi)$ that contains the those blocks.
The update algorithm in this section performs $\O(m+m')$ block-writes to EM, reads $\O(m'+ D' \cdot  m)$ blocks, and uses $\O(1)$ blocks of temporary MM storage during the update.
\end{observation}

Our rebuild-algorithm aims at minimizing the number of block-writes in an update, using a greedy top-down approach that builds the subtree of $T_{X'}$ in auxiliary EM storage while reading the keys in $T_X$ from UR EM.
On completion, we delete the $m$ obsolete blocks from UR memory and write the $m'$ result blocks back to UR to obtain the final tree $T_{X'}$.
This way, our update algorithm will only require $\O(1)$ blocks of temporary storage in MM (and still supports searches while rebuilding of the subtree is in progress).
The basic idea is as follows. 
Given a designated separator-interval $(\ell(v'), r(v'))$ when assembling block $v'$ for writing to auxiliary EM, 
we determine those keys with the $\alpha$-smallest priority values by searching in $T_X$, 
place them in the array of $v'$,
determine the fan-out $\delta_{v'}$ of and the active separators of block node $v'$, and 
recursively construct the remaining subtrees for those keys within each section between active separators of $v'$.
Eventually, we free the obsolete blocks of $T_X$ form UR and move the result subtree from the auxiliary to UR memory.
Next we discuss the possible rebuild cases when inserting a new key $x$ to the tree, their realization is discussed thereafter.
For a block node $v$, let $X(v)$ denote the set of keys stored in the subtree of $v$, $B(v)\subseteq X(v)$ the keys stored in the array of $v$ itself,  
$p^*(v)=\min \{ \pi(y):y \in B(v) \}$, and $p^\dagger(v)=\max\{\pi(y):y \in B(v)\}$.
(E.g. the key $y$ with $\pi(y)=p^*(v)$ is the UR label of the block $v$.)

Consider the search path for $x$ in the RBST $T_X$.
Starting at the root node, we stop at the first node $v$ that
% x is present
i)  must store $x$ in its array (i.e. $\pi(x)\in [p^*(v),p^\dagger(v)]$) or 
% x not present
ii) that requires rebuilding (i.e. must increase its fan-out).
To check for case ii), we use the augmented child references of $v$ to determine the number of keys in the result subtree $v'$.
Thus, both fan-outs, $\delta_v$ pre-insertion and $\delta_{v'}$ post-insertion, are known without additional block reads from the stored subtree weights\footnote{To avoid storing and updating subtree weights explicitly, the subtree weight can be computed bottom-up along the search path of $x$ until the node $v$, containing the $m'$ blocks, is found.}.

If neither i) nor ii) occurs, the search halts in a buffer list ($\delta_v =\delta_{v'}\leq 1$) and we simply insert the key and rewrite all subsequent blocks of the list in one linear pass.
Note that, if the list is omitted (as there are no keys in the active separator interval), the insertion must simply allocate a new block that solely stores $x$ to generate the output RBST.
It remains to specify the insertion algorithm for case i) and ii).
See Figure~\ref{fig:rebuild-cases} for an overview of all cases.

\begin{figure}[] \centering
    \includegraphics[width=\columnwidth]{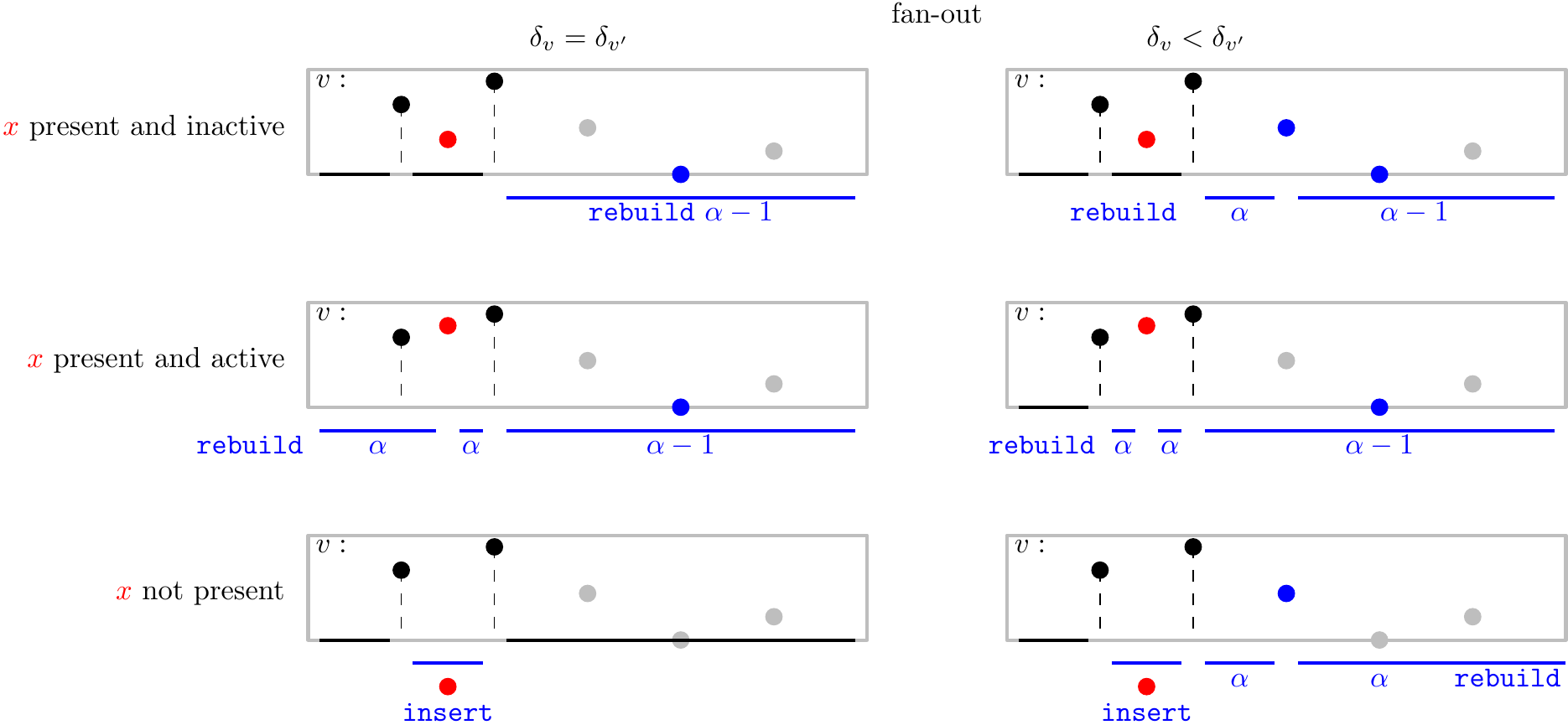}
    \caption{Rebuild cases for insertion of key $x$ (red) in block $v$ of an RBST.}
    \label{fig:rebuild-cases}
\end{figure}
 
In ii), $v$ is a buffer node that must increase its fan-out.
It suffices to build two trees from one subtree, each containing the top $\alpha$ keys in their root.
If the new key $x$ is contained in this interval, we pass $x$ as additional `carry' parameter to the procedure.
If $x$ is not contained in this interval, the insertion continues on the respective child of $v$ that must contain $x$. 
(See Figure~\ref{fig:rebuild-cases} bottom right.)
Note that at most three subtrees of $v'$ are built in auxiliary EM and that reads occur in at most two subtrees of $v$, regardless of the actual value of $\alpha$.

In i), $v$ is an internal node that stores $\alpha$ keys and there are four cases; depending on if $x$ is an active separator in $v'$ and if the fan-out stays or increases (see Figure~\ref{fig:rebuild-cases}). 
In all cases, it suffices to rebuild at most three subtrees.
Two trees for the two separator-intervals that are incident to $x$ and one tree that contains the two subtrees of the separator-intervals incident to the key $y$ that is displaced from block $v$ by the insertion of $x$, i.e. $\pi(y)=p^\dagger(v)$.
Note that deciding if insertion of $x$ in a block $v$ that is buffering leads to an increased fan-out does not require additional block reads, as the weight of the respective subtree is stored together with the child reference in $v$.

Note that deleting a key $x$ from an RBST has the same cases for rebuilding sections of keys between active separators.
Moreover, if a deletion leads to a buffering block ($\delta_{v'}<\alpha+1$), then all its children are buffering and thus store their subtree weight together with the reference to it. 
Thus, determining the fan-out for deletions also requires no additional block reads.

To complete our update algorithms, it remains to specify our top-down greedy rebuild procedure for a given interval between two active separators.
First, we isolate the task of finding the relevant keys, needed for writing the next block to auxiliary EM within the following range-search function \verb-top-$(k, v, (\ell,r))$.
Given a subtree root $v$ in $T_X$, a positive integer $k \leq \alpha$, and query range $(\ell,r)$, we report those keys from $Y = X(v) \cap (\ell,r)$ that have the $k$ smallest priority values and their respective subtree weights $w_0+\ldots+ w_k =|Y|$ based on a naïve search that visits (in ascending key order) all descendants of $v$ that overlap with the query range.

For an $\alpha$ or an $\alpha-1$ rebuild, we run the respective \verb-top- query on the designated interval range and
check if $x$ is stored in the output block or beneath it. 
Then we determine the fan-out from the range count results $\{w_i\}$, which determines the active separators in the output block.
We allocate one block in auxiliary EM for each non-empty child-reference and set their parent points.
Finally, we recursively rebuild the subtrees for the intervals of the active sections, passing them the reference to the subtree root in $T_X$ that holds all necessary keys for rebuilding its subtree.
Note that for this form of recursion, using parent pointers allows to implement the rebuild procedure such that only a constant number of temporary blocks of MM storage suffice.
After the subtrees are build in the auxiliary EM, we delete the $m$ obsolete blocks from UR and move the $m'$ blocks to UR to obtain the final tree $T'$.

Clearly, the number of block writes is $\O(m+m')$ and the number of block reads is $\O(D'\cdot m)$, from a simple charging argument.
That is, any one block $u$ of the $m$ blocks in old tree is only read by the \verb-top- call if the key range of $u$ intersects the query range $(\ell,r)$ of the \verb-top- call.
Consequently $(\ell,r)$ contains the smallest key, the largest key, or all keys of non-empty $u$.
The number of either such reads of $u$ is bounded by $D'$, since those blocks from the output tree that issue the call to \verb-top- have intervals that are contained in each other by the search tree property.
The expected I/O bounds of our partial rebuild algorithm will follow from our analysis of the depth and size of RBSTs in the next sections.

% \newpage
\section{Bounds for Searching and the Subtree Weight in RBSTs}
\label{sec:search-and-depth}
Since successful tree searches terminate earlier, it suffices to analyze the query cost for keys $q \notin X$, i.e. unsuccessful searches.
In this section, we bound the expected block reads for searching for some fixed $q$ in the block trees $T_X$, where the randomness is over the permutations $\pi \in \Perm(n)$, i.e. the set of bijections $\pi : X \to \{1,\ldots, n\}$.

Consider the sequence $v_1,v_2,\ldots$ of blocks of an RBST on the search path from the root to some fixed $q$.
Since the subtree weight is a (strictly) decreasing function, we have that the fan-out $\delta$ of the nodes is a non-increasing function that eventually assumes a value $\leq 1$.
From the definition of $\delta$, we have that there are (in the worst case) at most $O(1/\eps)$ blocks in the search path with $\delta = 1$.
Thus, it suffices to bound the expected number of internal blocks with $\delta \geq 2$, which all have the search tree property $[\ell(v_i),r(v_i)]\supset(\ell(v_{i+1}),r(v_{i+1})) \ni q$.
Our bound on the number of block reads of the range search in RBSTs will use the results of the next section and is thus presented at the end of it.

\begin{lemma}\label{lem:point-depth}
    Let $X$ be a set of $n$ keys, $q \in \mathbb{R} \setminus X$, and $2 \leq \alpha$ an integer.
    The expected value of random variable $D_q$, i.e. the number of primary tree nodes that have $q$ in their interval, is at most $ \E[D_q] \leq 5 \log_{\alpha} n $.
    The expected number of secondary tree notes (i.e. $\delta<\alpha+1$) that have $q$ in their interval is $\O(1/\eps)$.
    In particular, the expected number of blocks in the search path of $q$ is $\O( \eps^{-1} + \log_\alpha n)$.
\end{lemma}

The block structure of RBSTs do not allow a clear extension of Seidel's analysis of the Treap based on `Ancestor-Events'~\cite{SeidelA96} or the backward-analysis of point-location cost in geometric search structures~\cite[Chapter~6.4]{dutch-book}.
Our proof technique is vaguely inspired by a top-down analysis technique of randomized algorithms for \emph{static} problems that consider `average conflict list' size.
However, several new ideas are needed to analyze a dynamic setting where the location $q$ is user-specified and not randomly selected by the algorithm.
Our proof uses a partition in layers whose size increase exponentially with factor $\alpha$ and a bidirectional counting process of keys in a layer of that partition. 

\begin{figure} \centering
    \includegraphics[width=.5\textwidth]{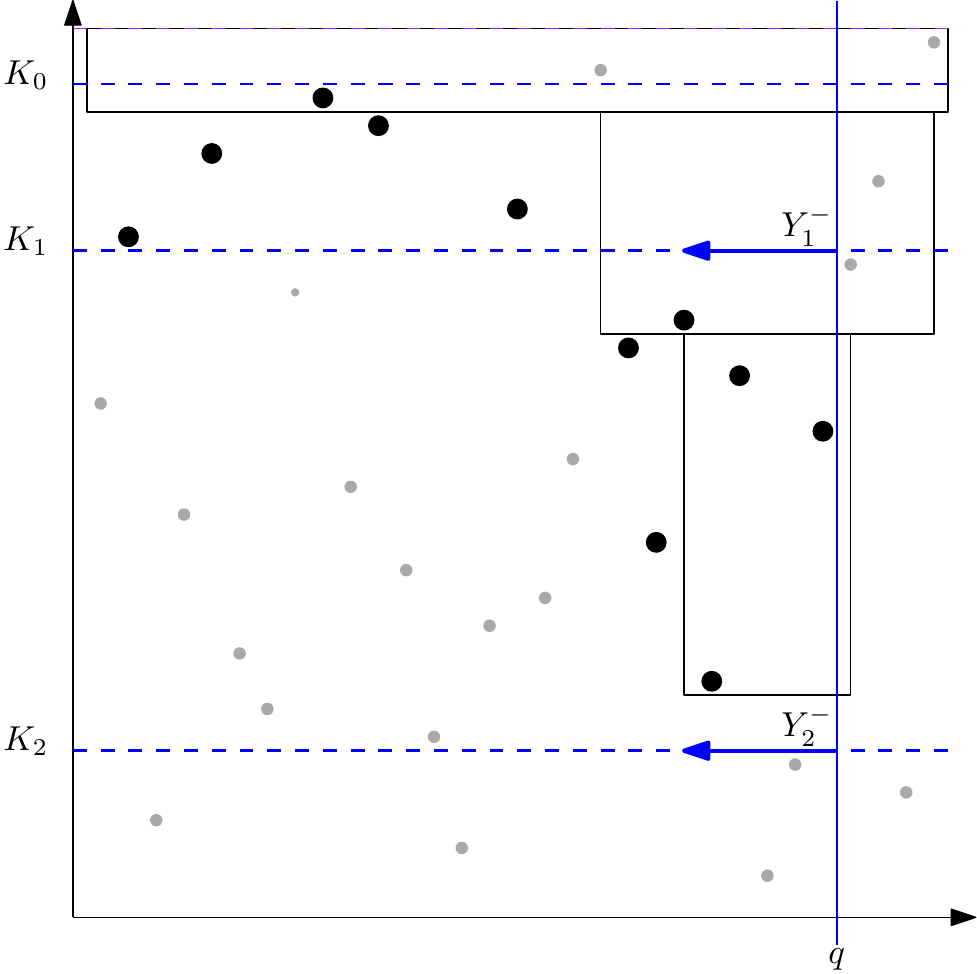}
    \caption{Illustration of the layer partitions $K_i$ and the number of keys counted by $Y_1^-=5$ and $Y_2^-=6$ (black) in the proof of Lemma~\ref{lem:point-depth}.}
    \label{fig:depth}
\end{figure}

\begin{proof}
Partition $\{1,\ldots,n\}$ with intervals of the form $[\alpha^i,\alpha^{i+1})$ for indices $0 \leq i \leq  \floor{ \log_\alpha n} $.
E.g. permutation $\pi$ induces an assignment of keys $x \in X$ to an unique layer index, i.e. $i$ with $\pi(x) \in [\alpha^i,\alpha^{i+1})$.
For internal node $v$ let $p^*(v)=\min \{ \pi(x):x\text{ is stored in }v \}$ and $p^\dagger(v)$ is the maximum over the set.
Note that $p^*(v) + \alpha -1 \le p^\dagger(v)$, since every primary node contains exactly $\alpha$ keys.
Thus, $D_q$ counts nodes $v$ that either have both $p^*(v),p^\dagger(v) \in [\alpha^i,\alpha^{i+1})$ or have that $p^\dagger(v)$ has a larger layer index than $i$.
We bound the expected number of nodes of the first kind, since there are in the worst case at most $\floor{\log_\alpha n}$ from the second kind.
Defining for every layer index $i$ a random variable $V_i$ that counts the tree nodes in that layer
$$
V_i(\pi) = 
\Big| 
\Big\{ 
    v \in T_X(\pi)~:~ \ell(v) < q < r(v) \text{ and } p^*(v),p^\dagger(v) \in [\alpha^i,\alpha^{i+1})
\Big\}
\Big| \quad,
$$
we have that 
$ D_q \leq \floor{\log_\alpha n} + V$, where $V = \sum_{i \ge 0} V_i$ is the total number of blocks of the first kind.
The reminder of the proof shows expectation bounds for $V_i$.

Let $K_i =\{ x \in X:\pi(x)<\alpha^{i+1} \}$ for each index~$i$, $K_i^- = \{x \in K_ i : x < q\}$, and $K_i^+=K_ i\setminus K_i^-$.
Thus, $K_i \subset K_{i+1}$, and $|K_i|=\alpha^{i+1}-1$.
Define $Y_i^-$ to be the number of consecutive keys form $K_i^-$ that are less than $q$ but not contained in $K_{i-1}$.
Analogously, random variable $Y_ i^+$ counts those larger than $q$ and 
$Y_i := Y^-_i + Y^+_i$ is the number of consecutive keys of $K_i$, whose range contain $q$, but do not contain elements from $K_{i-1}$.
Since all keys in a primary tree node are active separators, we have $V_i(\pi) \leq Y_i(\pi)/\alpha $ for every $\pi \in \Perm(n)$.

Next we bound $\E[Y_i^-]$ based on a sequence of binary events:
Starting at $q$, we consider the elements $x \in K_i^-$ in descending key-order and count as `successes' if $x \in K_i\setminus K_{i-1}$ until one failure ($x \in K_{i-1}$) is observed.
If $K_i^-$ has no more elements, the experiment continues on $K_i^+$ in ascending key order.
Defining $Z_i$ as the number of successes after termination, we have 
$Y_i^-(\pi) \leq Z_i(\pi)$.
The probability to obtain failure, after observing $j$ successes, is 
$\frac{|K_{i-1}|  }{|K_i|-j}$, which is at least $p_i:= \frac{|K_{i-1}|}{|K_i|}$ for all $j\geq 0$ and $i>0$.

Hence $\Pr[Z_i = j] \leq \Pr[Z'_i = j+1]$ where random variable $Z'_i\sim NB(1,p_i)$, i.e. the number of trials to obtain one failure in a sequence of independent Bernoulli trials with failure probability $p_i$.
Since $\E[Z'_ i]=1/p_i$, we have 
\begin{align}
    \E[Y_i^-] \leq \E[Z_i] < \E[Z'_i] = 1/p_i = \frac{\alpha^{i+1}-1}{\alpha^i-1}< \frac{\alpha}{1-1/\alpha^i}~.%\leq 2\alpha.
\end{align}
We thus have $\E[V_i] \leq \E[Y_i]/\alpha < \tfrac{2}{\alpha}\tfrac{\alpha}{1-1/\alpha^i} \leq 4$ for all $\alpha\geq 2$, which shows the lemma's statement for the primary tree nodes (that have $\delta=\alpha+1$).

Since there are in the worst case $\O(1/\eps)$ nodes with fan-out $\delta \leq 1$ on a search path, it remains to bound the expected number of secondary nodes with $2 \leq \delta(n') \leq \alpha$ on a search path, where $n' := |X(v)|$ is the subtree weight of top-most buffering node $v$ on the search path.
For any fixed $n'$, this random variable only depends on the relative order of the $n'$ keys, which are uniform from $\Perm(n')$.
Since the expected number of keys in either of the $\delta(n')$ sections is $(n'-\alpha)/\delta(n')$, the expected number of keys in section of $q$ has an upper bounded of the form $ (n'-\alpha)/\delta(n') =\O(\alpha/\eps)$.
Since the bound holds for all $n'$, the bound holds unconditional. 
Consequently, the lemma's $\O(1/\eps)$ expectation bound on the nodes follows from the fact that all secondary nodes (with $\delta\geq 2$) store exactly $\alpha$ keys.
\end{proof}

We will use this technique again in the analysis of the structural change in Section~\ref{sec:efficient-updates}.
From the Layer Partition in our previous proof, we obtain an upper bound on the expected weight of a subtree, subject to an update, of the form 
$\sum_i \frac{\alpha^{i+1}-\alpha^i}{n} \frac{n-\alpha^i+1}{\alpha^i-1} =\O(\alpha \log_\alpha n)$.
In the reminder of the paper, we will derive the tools to show that the expected structural change of an update, i.e. the number of block writes, is bounded within a $(1/\alpha^2)$-factor of this bound.

% \newpage
\section{Bounds on Size using a Top-Down Analysis}
\label{sec:size}

Our analysis will frequently use the following characterization of the partitions of a set $X$ of $n$ keys that are induced by the $\alpha$ elements of the smallest priority values from the set.

\begin{observation} \label{obs:Xi-characterization}
There is a bijective mapping between the partitions on $n$ keys, induced by the first $\alpha$ keys from $\pi \in \Perm(n)$, and the solutions to the equation
% \begin{align}
    $X_1 + X_2 + \dots + X_{\alpha+1} = n - \alpha $,
% \end{align}
 where the variables $X_i$ are non-negative integers.
 This bijection implies that the solutions of the equation happen with the same probability.
\end{observation}

In other word, $X_i$ is the number of keys in the $i$-th section beneath the root of $T_X$, where $1\leq i \leq \alpha+1$.
% \begin{proof}
Thus, $X_i$ can be considered as the number of consecutive keys from $X$ that are between the $i$-th and $(i+1)$-th key stored in the root.
% \end{proof}
For example, in an RBST on the keys $\{1,\ldots,10\}$ and block size $\alpha=3$, a root block consisting of the key values $4$, $7$, $8$ is characterized by the assignment $X_1=3,~ X_2=2,~ X_3=0,~ X_4=2$.

Next we analyze the effect of our secondary structures, since the size of RBSTs without buffer ($\beta=0$) can be dramatically large.
For example, the expected number of blocks $\E[S]$ for subtrees of size $n=2\alpha$ is 
\begin{align*}
    \E[S] &= 1 + \Pr[X_1>0] + \dots + \Pr[X_{\alpha+1}>0]\\
    &= 1 + (\alpha+1)\left( 1-\binom{n-1}{\alpha-1} / \binom{n}{\alpha} \right)
    = 1 + (\alpha+1)(1-\alpha/n)
    = 1 + (\alpha+1)/2 \quad.
\end{align*}
Note that in this example, buffering stores the whole subtree using only one additional block.

\subsection{Buffering for UR trees with load-factor $1-\eps$}
\label{sec:buffering}

Our top-down analysis of the expected size, and thus load-factors, of $(\alpha,\eps)$-RBSTs frequently uses the following counting and index exchange arguments.

\begin{lemma}[Exchange]\label{lem:index-exchange}
We have $\Pr [X_i < c] = \Pr[X_j < c]$ for all
$i,j \in \{1,\ldots,\alpha + 1\}$.
\end{lemma}
\begin{proof}
Due to Observation \ref{obs:Xi-characterization}, each solution to $X_1 + \dots + X_{\alpha+1}= n -\alpha$ occurs with the same probability. Hence, we can calculate $\Pr[X_i\geq c]$ by counting the number of solutions where $X_i\geq c$ and dividing it by the total number of solutions.
Thus
\begin{align*}
    \Pr[X_i < c] = 1- \Pr[X_i \geq c] = 1- \frac{\binom{n-c}{\alpha}}{\binom{n}{\alpha}}
    = 1 - \Pr[X_j \geq c]
    = \Pr[X_j < c]~,
\end{align*}
as stated.
\end{proof}

Next, we show an upper bound for the size of an $(\alpha,\eps)$-RBST with $n$ keys, where the priorities are from a uniform distribution over the permutations in $\Perm(n)$.
Our analysis crucially relies on the basic fact that restricting $\Perm(n)$ on an arbitrary key subset of cardinality $n'<n$ yields the uniform distribution on $\Perm(n')$.
Random variable $S_n$ denotes the space, i.e. the number of blocks used by the RBST on a set of $n$ keys, $F_n$ denotes the number of full blocks, and $E_n := S_n - F_n$ the number of non-full blocks. 
Next we show a probability bound for the event that a given section, say $k$, contains a number of keys $X_k$ of a certain range.

\begin{lemma}\label{lem:range-nof-keys}
For any $1\leq  k\leq \alpha+1$ and $i<j$, we have
% \begin{align}
   $ 
    \sum_{x=i}^{j} \Pr[X_k=x] = \frac{\binom{n-i}{\alpha}-\binom{n-1-j}{\alpha}}{\binom{n}{\alpha}}
    $.
% \end{align}
\end{lemma}

\begin{proof}
    We have $\Pr[X_k\geq i] = \binom{n-i}{\alpha} / \binom{n}{\alpha}$ and $\Pr[X_k\geq j+1] = \binom{n-(j+1)}{\alpha} / \binom{n}{\alpha}$.
\end{proof}

These basic facts allow us to compute the following expressions.

\begin{lemma} \label{lem:conditional-expected}
    We have $\sum_{x=1}^{m} \Pr[X_k=x] x  = \left( \binom{n}{\alpha+1} - \binom{n-m}{\alpha+1} - m\binom{n-1-m}{\alpha} \right) / \binom{n}{\alpha}$.
\end{lemma}

\begin{proof}
By elementary computation, we have
\begin{align}
\sum_{x=1}^{m} \Pr[X_k=x] x 
&= \sum_{x=1}^{m} \Pr[X_k=x] + \sum_{x=2}^{m} \Pr[X_k=x] + \dots  + \sum_{x=m}^{m} \Pr[X_k=x]  \nonumber \\
&= \frac{\binom{n-1}{\alpha} - \binom{n-1-m}{\alpha}}{\binom{n}{\alpha}} + \ldots + \frac{\binom{n-m}{\alpha} - \binom{n-1-m}{\alpha}}{\binom{n}{\alpha}}
\nonumber \\
&= \sum_{i=1}^{m}\frac{\binom{n-i}{\alpha}}{\binom{n}{\alpha}} - m\frac{\binom{n-1-m}{\alpha}}{\binom{n}{\alpha}}~. \label{eq:cexpec-a}
\end{align}
Note that an RBST on $n+1$ keys with block size $\alpha+1$ has $\Pr[X_1=i] = \frac{\binom{n-i}{\alpha}}{\binom{n+1}{\alpha+1}}$. 
Thus, we have from Lemma~\ref{lem:range-nof-keys} that 
\begin{align}
    \sum_{i=1}^{m}\Pr[X_1=i] = \sum_{i=1}^{m} \frac{\binom{n-i}{\alpha}}{\binom{n+1}{\alpha+1}} = \frac{\binom{(n+1)-1}{\alpha+1} - \binom{(n+1)-1-m}{\alpha+1}}{\binom{n+1}{\alpha+1}}\\
    \Rightarrow
    \sum_{i=1}^{m} \binom{n-i}{\alpha} = \binom{n}{\alpha+1} - \binom{n-m}{\alpha+1}
    \quad \label{eq:cexpec-b}
\end{align}
and the lemma follows by using (\ref{eq:cexpec-b}) for the summation in (\ref{eq:cexpec-a}).
\end{proof}

We are now ready to prove our size bound.
\begin{theorem}[Size] \label{lem:size}
The expected number of non-full blocks in an $n$ key $(\alpha,\eps)$-RBST, $\E[E_n]$, is at most $\max \{ \eps n/\alpha,~1\} $.
\end{theorem}

Our proof is by induction on $n$, where the base cases are due to bounds on the number of non-full blocks that are occupied by the secondary buffer structures (see Appendix~\ref{sec:size-buffer}).
\begin{proof}
    Observation \ref{obs:buffer-delta-one-nonfull-blocks} states that $E_n\leq 1$ for all $n$ with $\delta_n=1$.
    Moreover, Theorem \ref{thm:buffer-size-upper-bound} shows that $\E[E_n]\leq 27 \delta_n$ for $n \leq \beta+\alpha = (\alpha+1) \rho+\alpha$.
    For simplicity, we will use $t$ and $\gamma$ instead of $\rho+\alpha$ and $(\alpha+1)\rho+\alpha$ respectively.
    We can conclude from Theorem \ref{thm:buffer-size-upper-bound} that for $t < n \leq \gamma$,
    \begin{align*}
        \E[E_n]     &\leq 27 \left\lceil \frac{n-\alpha}{\rho} \right\rceil
                        = 27 \left\lceil \frac{n-\alpha}{108\alpha/\eps} \right\rceil\\
                    &\leq 27 \left(\frac{n-\alpha}{108\alpha/\eps}+1 \right)
        \leq 27\cdot2\frac{n-\alpha}{108\alpha/\eps} = \frac{n-\alpha}{2\alpha/\eps}~.
    \end{align*}
    The last inequality holds because $n>\rho+\alpha=108\alpha/\eps+\alpha$ and $\frac{n-\alpha}{108\alpha/\eps}>1$.\\
    
    For $n > \gamma$, we prove the theorem by induction. 
    For each index $i \in \{1,\ldots, \alpha+1\}$, define $Y_{i}$ as the number of non-full blocks in the $i$-th section beneath the root (of an RBST with $n$ keys).
    Note that $Y_{i}$ only depends on $X_i$, i.e. the number of keys in the $i$-th section, and their relative priorities.
    We thus have 
    
    \begin{align} 
        \E[Y_i] &= \sum_{x_i=1}^{n-\alpha} \Pr[X_i=x_i] \cdot \E [E_{x_i}] 
        \nonumber\\
        &\leq \sum_{x_i=1}^{t} \Pr[X_i=x_i] \cdot 1 + \sum_{x_i=t+1}^{\gamma} \Pr[X_i=x_i]\frac{x_i}{2\alpha/\eps} + \sum_{x_i=\gamma+1}^{n-\alpha} \Pr[X_i=x_i] \frac{x_i}{\alpha/\eps}\\
        &= \left( \sum_{x_i=1}^{t} \Pr[X_i=x_i] 
        + \sum_{x_i=t+1}^{\gamma} \Pr[X_i=x_i] \right) - \sum_{x_i=t+1}^{\gamma} \Pr[X_i=x_i] \nonumber\\
        &+ \left(\sum_{x_i=t+1}^{\gamma} \Pr[X_i=x_i]\frac{x_i}{2\alpha/\eps} 
        - \sum_{x_i=t+1}^{\gamma} \Pr[X_i=x_i] \frac{x_i}{\alpha/\eps}\right)
        -2 \sum_{x_i=1}^{t} \Pr[X_i=x_i] \frac{x_i}{2\alpha/\eps} \nonumber\\
        &+ \left(\sum_{x_i=1}^{t} \Pr[X_i=x_i] \frac{x_i}{\alpha/\eps}
        + \sum_{x_i=t+1}^{\gamma} \Pr[X_i=x_i] \frac{x_i}{\alpha/\eps}
        + \sum_{x_i=\gamma+1}^{n-\alpha} \Pr[X_i=x_i] \frac{x_i}{\alpha/\eps} \right)\\
        &= \left(\sum_{x_i=1}^{\gamma} \Pr[X_i=x_i] 
        -\sum_{x_i=t+1}^{\gamma} \Pr[X_i=x_i]\frac{x_i}{2\alpha/\eps} 
        - \sum_{x_i=1}^{t} \Pr[X_i=x_i] \frac{x_i}{2\alpha/\eps} \right) \nonumber\\
        &+ \sum_{x_i=1}^{n-\alpha} \Pr[X_i=x_i] \frac{x_i}{\alpha/\eps}
        - \sum_{x_i=1}^{t} \Pr[X_i=x_i] \frac{x_i}{2\alpha/\eps}
        - \sum_{x_i=t+1}^{\gamma} \Pr[X_i=x_i]
        \\
        &\leq \sum_{x_i=1}^{\gamma} \Pr[X_i=x_i] \left(1-\frac{x_i}{2\alpha/\eps} \right) + \frac{\E[X_i]}{\alpha/\eps}
        \end{align}
        
        Using Lemma~\ref{lem:index-exchange}, we have that the summation term has equal value for each section $i \in \{1,\ldots, \alpha+1\}$. 
        Since, for $n \ge \alpha$, the root is full and not counted in $E_{n}$, we have
        \begin{align} 
         \E[E_n]  = \E[Y_1+ ... + Y_{\alpha+1} ]
          &\leq (1 +\alpha) \sum_{x_1=1}^{\gamma} \Pr[X_1=x_1] \left( 1-\frac{x_1}{2\alpha/\eps} \right) + \frac{n-\alpha}{\alpha/\eps}~. \label{eq:rhs-induction}
        \end{align}

    To proof the inequality it suffices to show that the right-hand side of (\ref{eq:rhs-induction}) is at most $\frac{\eps n}{\alpha}$.
    This is true if and only if 
    $ (1 + \alpha) \sum_{x_1=1}^{\gamma} \Pr[X_1=x_1] \left( 1-\eps\frac{x_1}{2\alpha} \right) \leq \eps $.
    Since $\eps >0$ and $\alpha+1 > 0$, it suffices to show that the value of the sum is not positive. This is if and only if
    \begin{align}
        \sum_{x_1=1}^{\gamma} \Pr[X_1=x_1]
        \le \frac{\eps}{2\alpha}\sum_{x_1=1}^{\gamma} \Pr[X_1=x_1]x_1~.
    \end{align}
    Use Lemma~\ref{lem:range-nof-keys} on the left-hand side and Lemma~\ref{lem:conditional-expected} on the right-hand side, it remains to show that
    \begin{align}
    \frac{2\alpha}{\eps}\binom{n-1}{\alpha} - \frac{2\alpha}{\eps}\binom{n-1-\gamma}{\alpha}
    &\le \binom{n}{\alpha+1} - \binom{n-\gamma}{\alpha+1} - \gamma\binom{n-1-\gamma}{\alpha}~.
    \end{align}
    By elementary computation, this inequality holds for all $n>\gamma$.
    (See Appendix~\ref{sec:final-claim} Lemma~\ref{lem:final-claim} for a full proof of this claim.)
    \end{proof}

Since each full block contains $\alpha$ distinct keys, i.e. $F_n\le n/\alpha$, we showed that $S_n\le (1+\eps)n/\alpha$.
Thus, we obtain for the load-factor $L$, i.e. the relative utilization of keys in allocated blocks, the following lower bound.

\begin{corollary}[Load-Factor]\label{lem:total-size}
Let $\eps \in \left(0,1/2\right]$ and $n\geq \alpha +\beta$.
The expected number of blocks occupied by an {$(\alpha,\eps)$-RBST} on $n$ keys is at most $(1+\eps)n/\alpha$, i.e. the expected load-factor is at least $1-\eps$.
\end{corollary}

\begin{proof}
The load-factor is the random variable
$L = n/\alpha S_n$. 
Since $\phi(x)=1/x$ is a convex function, Jensen's inequality gives 
$1/\E[X]\le \E[1/X]$. 
Using Lemma~\ref{lem:size}, we have
$    
\E[L] \geq \frac{n}{\alpha}\frac{1}{(1+\eps)n/\alpha} = \frac{1}{1+\eps} \geq 1 - \eps$.
\end{proof}

Combining the results of Lemma~\ref{lem:point-depth} and Theorem~\ref{lem:size}, we showed the following bound for reporting all results of a range-search in RBSTs.
\begin{corollary}[Range-Search]
    Let $\eps \in (0,1/2]$. 
    The expected number of block reads to report all results of a range search in $(\alpha,\eps)$-RBSTs is at most $\O(\eps^{-1} + k/\alpha + \log_\alpha n)$, where $k$ is the number of result keys.
\end{corollary}

We are now ready to show our bound of the write-efficiency of $(\alpha,\eps)$-RBSTs.

% \newpage
\section{Bounds for Dynamic Updates based on Partial-Rebuilding}
\label{sec:efficient-updates}
The weight-proportional fan-out in the design of our secondary UR tree structure has the advantage, in comparison to the use of a simple list of $O(\beta/\alpha)$ blocks, that it avoids the need of additional restructuring work whenever a subtree must change its state between buffering and non-buffering, due to updates (cf. Section~\ref{sec:algo-updates}).
Together with the space bound from last section and the observation that the partial-rebuild algorithm for updating RBSTs only rebuilds at most $3$ subtrees beneath the affected node, regardless of its fan-out, allows us to show our bound on the expected structural change in this section.

\begin{theorem}\label{lem:update-bound}
    Let $\eps \in (0,1/2]$. 
    The expected total structural change of an update in an $n$ key $(\alpha,\eps)$-RBST is $\O\left( \frac{1}{\eps} + \frac{\log_\alpha n}{\alpha}\right)$, i.e. $\O(1/\eps)$ for $\alpha = \Omega(\log(n)/\log \log (n))$.
\end{theorem}

\begin{proof}
It suffices to bound the expectation of the total structural change for deleting a key $x$ from the RBST, since the case of insertion of the key would count the same number of blocks.
Let $X'=X\cup\{x\}$ be the set of keys before and after the deletion of $x$ respectively, where $|X'|=n$.
For the subtree root $v$ that is subject to the partial rebuild algorithm from Section~\ref{sec:algo-updates}, there are three cases, i.e. $\delta_v \leq 1$, $\delta_v \in [2,\alpha]$, and $\delta_v=\alpha+1$.

For $\delta_v\leq 1$, we have in the worst case at most $\O(1/\eps)$ blocks in the buffer's list that are modified.

For $\delta_v=\alpha+1$, $v$ is a node of the primary tree and we consider the the layer partition of the priority values $\{1,\ldots,n\}$ in intervals of the form $[\alpha^i, \alpha^{i+1})$ for integer $i<\log_\alpha n$.
The probability that the priority of key $x$ falls in the interval of layer $i$ is $\Pr[\pi(x)\in [\alpha^i,\alpha^{i+1})]=\frac{\alpha^{i+1}-\alpha^i}{n} \leq \alpha^{i+1}/n$.
For layer $i$, let $K=\{x' \in X~:~\pi(x')<\alpha^i\}$ be the set of separators that partition the keys with priorities $\geq \alpha^i$, let random variable 
$Y(\pi)$ be the number of keys of $X\setminus K$ that are in the same section as $x$.
Thus, for all permutations $\pi$, the weight of the subtree of $v$ is at most $|X(v)| \leq Y$.
Using the characterization of Observation~\ref{obs:Xi-characterization}, we also have that total expected weight of three, from the $\alpha+1$, subtrees of $v$ is at most $\frac{3}{\alpha+1}Y$.
Clearly, $\E[Y]=\sum_{n'}n'\Pr[Y = n']$ by definition and we have $\E[Y]=\O(\frac{n-\alpha^i}{\alpha^i})=\O(n/\alpha^i)$ from direct calculation.
Note that the key count $n'$ per section is regardless of their relative order, i.e. each of the $(n')!$ orders is equally likely.
Thus, Theorem~\ref{lem:size} implies that the expected number of blocks in an RBST on $\frac{3}{\alpha+1}n'$ keys is bounded within a $\frac{1+\eps}{\alpha}$ factor. 
Thus, for $x$ in layer $i$, the expected number of blocks $m'$ in the three subtrees of $v$ is at most
\begin{align}
  \E[ ~m'~|~\pi(x) \in [\alpha^i,\alpha^{i+1})~] 
&=
    \sum_{n'} \E[ ~m'~|~\pi(x) \in [\alpha^i,\alpha^{i+1})~,~Y=n'~]\Pr[Y=n'] \\
&\leq   \sum_{n'} \frac{1+\eps}{\alpha} \left(\frac{3}{\alpha+1}n'\right) \Pr[Y=n'] \label{eq:bound-update}\\
&=      3\frac{1+\eps}{\alpha(\alpha+1)}\E[Y] =\O(n/\alpha^{i+2})~.
\end{align}
Consequently, the expected number of block writes is at most
\begin{align*}
    \E[m'] 
&= 
\sum_{i=0}^{\lfloor \log_\alpha n \rfloor} \E\left[~m'~|~\pi(x)\in [\alpha^i,\alpha^{i+1})\right] \Pr\left[\pi(x)\in [\alpha^i,\alpha^{i+1}) \right] 
\\
&\leq 
\sum_i \O\left(\frac{n}{\alpha^{i+2}}\right)\frac{\alpha^{i+1}}{n} 
= \O\left(\frac{\log_\alpha n}{\alpha}\right)\quad.
\end{align*}

For $\delta_v \in [2,\alpha]$, $v$ is a node of the secondary tree and we have its subtree weight $n'$ that $n' =\Theta(\delta_v \alpha / \eps)$.
Thus, Eq.~(\ref{eq:bound-update}) is $\sum_{n'} \frac{1+\eps}{\alpha}\frac{3}{\delta_v} \Theta(\delta_v \alpha/\eps) \Pr[Y=n'] = \O(1/\eps)$ for this case.
\end{proof}

\bibliography{../refs.bib}

\appendix

\section{Induction Base Case: Expected Size of Buffering Trees}
\label{sec:size-buffer}

\begin{lemma}
\label{lem:lower-upper-on-accumulative-prob}
Let $X_1,\dots,X_m$ be some non-negative integral random variables  where $X_1+X_2+\dots+X_m=n$. For integers $n\geq 1$ and $t\leq n$, we  have the tail-bounds
\begin{align*}
    \left(1-\frac{t}{n+1}\right)^{m-1}~~\leq~~ \Pr[X_i \geq t] ~~\leq~~ \left(1-\frac{t}{n+m-1}\right)^{m-1}
\end{align*} 
\end{lemma}

\begin{proof}
\begin{align*}
    \Pr[X_i\geq t] &= 
    \frac{\binom{n-t+m-1}{m-1}}{\binom{n+m-1}{m-1}} =\frac{(n-t+m-1)!~(n-t)!~(m-1)!}{(n+m-1)!~n!~(m-1)!}\\
    &= \frac{n-t+m-1}{n+m-1} ~\frac{n-t+m-2}{n+m-2} ~\dots~ \frac{n-t+m-(m-1)}{n+m-(m-1)}\\
    &=\left(1-\frac{t}{n+m-1}\right) \left(1-\frac{t}{n+m-2}\right) \dots \left(1-\frac{t}{n+m-(m-1)}\right)
\end{align*}
Since for $1\leq i\leq m-1$, $\frac{t}{n+m-1}\leq \frac{t}{n+m-i}\leq \frac{t}{n+m-(m-1)}=\frac{t}{n+1}$, the lemma holds.
\end{proof}

\begin{lemma}
    \label{lem:subtracting-same-amount-from-num-and-denom}
    For $0< c\leq a\leq b$, we have  $\frac{a-c}{b-c}\leq \frac{a}{b}$.
\end{lemma}
\begin{proof}
    $
       \frac{a-c}{b-c}\leq \frac{a}{b} \Longleftrightarrow ab-bc\leq ab-ac \Longleftrightarrow a\leq b
    $.
\end{proof}

For a subtree of size $n\leq \beta+\alpha=:(\alpha+1)\rho+\alpha$, fan-out parameter $\delta_n=\min\{\alpha+1,\max\{1, \lceil\frac{n-\alpha}{\rho}\rceil\}\}$ is equal to $\max\{1, \lceil \frac{n-\alpha}{\rho}\rceil\}$.
Therefore, for $n$ with $2\leq\delta_n\leq\alpha+1$, we have $ (\delta_n-1)\rho+\alpha <n\leq \delta_n\rho+\alpha$.
Moreover, for $k\in\{2,3,\dots,\delta_n\}$, $n>(k-1)\rho+\alpha$ if and only if $\delta_n\geq k$.

\begin{lemma}
\label{lem:upper-bound-on-accumulative-prob}
For a subtree of size $n\leq\beta+\alpha$, where $\delta_n\leq \alpha+1$, and all $k\in\{2,3,\dots,\delta_n\}$, we have
\begin{align*}
    \Pr[\delta_{X_i}\geq k] = \Pr[X_i>(k-1)\rho+\alpha] \leq \left(1-\frac{k-1}{\delta_n}\right)^{\delta_n-1}~,
\end{align*}
where $X_i$ is the number of keys in the i-th section beneath the root.
\end{lemma}

\begin{proof}
    \begin{align*}
        \Pr[X_i>(k-1)\rho+\alpha] &= \Pr[X_i\geq(k-1)\rho+\alpha+1]
    \end{align*}
    Set variables $m$, $t$, and $n$ of Lemma~\ref{lem:lower-upper-on-accumulative-prob} to $\delta_n$, $(k-1)\rho+\alpha+1$, and $n-\alpha$ respectively.
    We have
    \begin{align*}
        \Pr[X_i\geq(k-1)\rho+\alpha+1] 
        \leq 
        \left(1-\frac{(k-1)\rho+\alpha+1}{n-\alpha+\delta_n-1}\right)^{\delta_n-1}~.
    \end{align*}

    Since $n\leq \delta_n \rho+\alpha$, we have
    \begin{align*}
        \Pr[X_i > (k-1)\rho+\alpha]
        \leq 
        \left(1-\frac{(k-1)\rho+\alpha+1}{\delta_n \rho+\alpha-\alpha+\delta_n-1} \right)^{\delta_n-1}
    \end{align*}
    Using Lemma~\ref{lem:subtracting-same-amount-from-num-and-denom}, we subtract $\delta_n-1$ from the numerator and denominator of $\frac{ (k-1)\rho+\alpha+1 }{ \delta_n \rho+\alpha-\alpha+\delta_n-1}$.
Thus
    \begin{align*}
        \Pr[X_i>(k-1)\rho+\alpha]
        \leq 
        \left(1- \frac{(k-1)\rho+\alpha+1-(\delta_n-1)}{\delta_n \rho} \right)^{\delta_n-1}~.
    \end{align*}
    Since $\delta_n$ is at most $\alpha+1$, we have $\alpha+1-(\delta_n-1)$ is positive and conclude
    \begin{align*}
        \Pr[X_i>(k-1)\rho+\alpha] \leq 
        \left(1-\frac{(k-1)\rho}{\delta_n \rho} \right)^{\delta_n-1} 
        = 
        \left(1-\frac{k-1}{\delta_n} \right)^{\delta_n-1}~,
    \end{align*}
    as stated.
 \end{proof}

We will combine the results of Lemmas \ref{lem:upper-bound-on-accumulative-prob} and \ref{lem:subtracting-same-amount-from-num-and-denom} and get the following result.
\begin{corollary}
\label{simplified-exp-lem:upper-bound-on-accumulative-prob}
    For $2\leq k\leq \delta_n$, we have
    \begin{align*}
        \Pr[\delta_{X_i}\geq k] &= \Pr[X_i>(k-1)\rho+\alpha]
        \leq \left(1-\frac{k-1-1}{\delta_n-1}\right)^{\delta_n-1} \leq e^{-(k-2)} ~.
    \end{align*}
\end{corollary}

\begin{observation}
\label{obs:buffer-delta-one-nonfull-blocks}
For a subtree with $\delta_{n}=1$, the data structure is a simple list and has at most one non-full block.
\end{observation}

\begin{lemma}
\label{lem:buffer-delta-two-nonfull-blocks}
For a subtree with $\delta_n =2$,
% with $\rho+\alpha < n \leq 2\rho+\alpha$ keys,
 the expected number of non-full blocks $\E[E_n]\leq 5$.
\end{lemma}

\begin{proof}
    We will design an algorithm calculating an upper bound on the number of non-full blocks.
    The root has $\alpha$ keys, and the remaining keys are randomly split into two sections with $X_1$ and $X_2$ keys, i.e. $X_1+X_2=n-\alpha$.
    For $i\in\{1,2\}$, if $\delta_{X_i}=1$, there is at most one non-full block, and the algorithm does not need to proceed in this section anymore.
    If $\delta_{X_i}=2$, to observe non-full blocks, the $X_i-\alpha$ keys should be partitioned again.
    It is impossible for both $X_1$ and $X_2$ to have $\delta_{X_i} = 2$ since it implies that $n=\alpha+X_1+X_2\geq \alpha+2(\rho+\alpha) >2\rho+3\alpha$, which is a contradiction.
    So at each iteration, either the algorithm stops or continues with one of the sections.
    In the last step, there are two sections each with at most one non-full block.
    Thus, the expected number of iterations plus $1$ is an upper bound on $\E[E_n]$.

    Clearly, $\E[E_n]$ is indeed an increasing function of n, so losing $\alpha$ keys of the root at each iteration results in fewer non-full blocks.
    To obtain a weaker upper bound, we assume that at each step, the key with the highest priority splits the $n-1$ remaining keys, instead of $n-\alpha$ keys, into two new sections.
    These keys include the $\alpha-1$ other keys of the root as well.
    It suffices to show that the expected number of  iterations is at most $4$.

    At iteration $k$, $k$ keys with the highest priorities split $n$ keys into $k$ sections with $Y^{(k)}_1, Y^{(k)}_2,\dots, Y^{(k)}_k$ keys, where $Y^{(k)}_1+Y^{(k)}_2+\dots+Y^{(k)}_k=n-k$.
    Round $k+1$ occurs if any of $Y^{(k)}_i$s exceeds $\rho+\alpha$.
    Since $k$ top keys are chosen uniformly at random, all solutions to the equation $Y^{(k)}_1+Y^{(k)}_2+\dots+Y^{(k)}_k=n-k$ have equal probabilities.
    Next, we find an upper bound on $\Pr[Y^{(k)}_i>\rho+\alpha]$ for $i\in\{1,\dots,k\}$.
    Set parameters $m$, $t$, and $n$ of Lemma~\ref{lem:lower-upper-on-accumulative-prob} to $k$, $\rho+\alpha$, and $n-k$ respectively.
    Thus
    \begin{align*}
        \Pr[Y^{(k)}_i>\rho+\alpha] &\leq 
        \left(1-\frac{\rho+\alpha+1}{n-k+k-1} \right)^{k-1}\\
        &\leq \left(1-\frac{\rho+\alpha-1}{n-1}\right)^{k-1}~.
    \end{align*}
    Since $n$ is at least $2\rho+\alpha$, we have
    \begin{align*}
        \Pr[Y^{(k)}_i >\rho+\alpha] &\leq \left(1-\frac{\rho+\alpha-1}{2\rho+\alpha-1}\right)^{k-1}\\
        &\leq \left(1-\frac{\rho}{2\rho}\right)^{k-1} = 2^{-(k-1)}~.
    \end{align*}
    The last inequality is an application of Lemma~\ref{lem:subtracting-same-amount-from-num-and-denom}.
    Using the union bound, we have
    \begin{align*}
        \Pr[\exists Y^{(k)}_i>\rho+\alpha] &\leq k / 2^{k-1}~.
    \end{align*}
    Define $Z$ to be the number of iterations and $Z_k$ to be an indicator random variable of whether the $k$-th iteration happened.
    It follows that
    \begin{align*}
        \E[Z] &=\sum_{k=1}^{\infty} \E[Z_k] =\sum_{k=1}^{\infty} \Pr[Z_k=1]\\
        &\leq \sum_{k=1}^{\infty} k/2^{k-1} 
        = 2(\sum_{j=1}^{\infty}\sum_{k=j}^{\infty}2^{-k})
        =2(\sum_{j=1}^{\infty}\frac{2^{-j}}{1-\frac{1}{2}})
        =2\frac{\frac{2^{-1}}{1-1/2}}{1/2} = 4 ~.
    \end{align*}
\end{proof}

\begin{lemma}
\label{lem:buffer-delta-three-nonfull-blocks}
For a subtree with $\delta_n=3$,
% with $2\rho + \alpha < n \leq 3\rho+\alpha$ keys, 
the expected number of non-full blocks $\E[E_n] \leq 10$.
\end{lemma}

\begin{proof}
    $n-\alpha$ keys beneath the root are separated into three sections with $X_1, X_2,$ and $X_3$ keys.
    Same as Lemma \ref{lem:buffer-delta-two-nonfull-blocks}, we can prove that it is  impossible to have $\delta_{X_i}\geq 2$ for every $i\in\{1,2,3\}$, or $\delta_{X_i}=\delta_{X_j}=3$ for some $i\neq j\in \{1,2,3\}$.
    Due to the symmetric property of the random variables $X_i$, we can conclude
    \begin{align*}
        1 =&~~
        3\Pr[\delta_{X_1}=3\wedge\delta_{X_2}=1\wedge\delta_{X_3}=1]\\
        &+3\Pr[\delta_{X_1}=3\wedge\delta_{X_2}=2\wedge\delta_{X_3}=1]\\
        &+3\Pr[\delta_{X_1}=2\wedge\delta_{X_2}=1\wedge\delta_{X_3}=1]\\
        &+\Pr[\delta_{X_1}=1\wedge\delta_{X_2}=1\wedge\delta_{X_3}=1]~.
    \end{align*}

    Next, we would do induction on $2\rho+\alpha < n\leq 3\rho+\alpha$.
    Using Observation \ref{obs:buffer-delta-one-nonfull-blocks}, Lemma \ref{lem:buffer-delta-two-nonfull-blocks}, and induction hypothesis, we get
        \begin{align*}
        \E[E_n] \leq& ~~3\Pr[\delta_{X_1}=3\wedge\delta_{X_2}=1\wedge\delta_{X_3}=1] (10+1+1)\\
        &+3\Pr[\delta_{X_1}=2\wedge\delta_{X_2}=2\wedge\delta_{X_3}=1] (5+5+1)\\
        &+3\Pr[\delta_{X_1}=2\wedge\delta_{X_2}=1\wedge\delta_{X_3}=1] (5+1+1)\\
        &+\Pr[\delta_{X_1}=1\wedge\delta_{X_2}=1\wedge\delta_{X_3}=1] (1+1+1)
        \\
        <& ~~3\Pr[\delta_{X_1}=3\wedge\delta_{X_2}=1\wedge\delta_{X_3}=1](12-7)\\
        &+ 3\Pr[\delta_{X_1}=2\wedge\delta_{X_2}=2\wedge\delta_{X_3}=1] (11-7)\\
        &+ 7(3\Pr[\delta_{X_1}=3\wedge\delta_{X_2}=1\wedge\delta_{X_3}=1]\\
        &+ 3\Pr[\delta_{X_1}=2\wedge\delta_{X_2}=2\wedge\delta_{X_3}=1]\\
        &+3\Pr[\delta_{X_1}=2\wedge\delta_{X_2}=1\wedge\delta_{X_3}=1]\\
        &+\Pr[\delta_{X_1}=1\wedge\delta_{X_2}=1\wedge\delta_{X_3}=1])
        \\
        =&~~ 3\Pr[\delta_{X_1}=3\wedge\delta_{X_2}=1\wedge\delta_{X_3}=1] 5 + 3\Pr[\delta_{X_1}=2\wedge\delta_{X_2}=2\wedge\delta_{X_3}=1] 4 + 7\\
        &\leq 7 + 3\Pr[\delta_{X_1}\geq 3] 5 + 3\Pr[\delta_{X_1}\geq 2\wedge\delta_{X_2}\geq 2] 4~.
    \end{align*}

    Corollary \ref{simplified-exp-lem:upper-bound-on-accumulative-prob} implies that $\Pr[\delta_{X_1}\geq 3]\leq \frac{1}{9}$.
    To complete the proof we show that $\Pr[\delta_{X_1}\geq 2\wedge\delta_{X_2}\geq 2]\leq \frac{1}{9}$.
    Consequently, we get the desired inequality $\E[E_n]\leq 7+3\frac{5}{9}+3\frac{4}{9}=10$.\\

    The two top-priority keys split the keys into three parts with $Y_1$, $Y_2$, and $Y_3$ keys, where $Y_1+Y_2+Y_3=n-2$.
    The i-th part includes all $X_i$ keys beneath the root and some extra keys from the root, so for each $i\in \{1,2,3\}$, we have $X_i\leq Y_i$.
    \begin{align*}
        \Pr[\delta_{X_1}\geq 2\wedge\delta_{X_2}\geq 2] &= \Pr[X_1\geq  \rho+\alpha+1 \wedge X_2\geq \rho+\alpha+1]\\
        &\leq \Pr[Y_1\geq \rho+\alpha+1 \wedge Y_2\geq \rho+\alpha+1]\\
        &= \frac{\binom{n-2(\rho+\alpha+1)}{2}}{\binom{n}{2}}
    \end{align*}
    Same technique of Lemma \ref{lem:lower-upper-on-accumulative-prob} yields
    \begin{align*}
        \frac{\binom{n-2(\rho+\alpha+1)}{2}}{\binom{n}{2}} &\leq 
        \left(1-\frac{2(\rho+\alpha+1)}{n} \right)^2~.
    \end{align*}
    Combining the previous inequalities together with $n\leq 3\rho+\alpha$ implies
    \begin{align*}
        \Pr[\delta_{X_1}\geq 2\wedge\delta_{X_2}\geq 2] 
        &\leq 
        \left(1-\frac{2(\rho+\alpha+1)}{3\rho+\alpha} \right)^2 ~.
    \end{align*}
    Subtract $\alpha$ from the numerator and denominator.
    By Lemma \ref{lem:subtracting-same-amount-from-num-and-denom}, we have:
    \begin{align*}
        \Pr[\delta_{X_1}\geq 2\wedge\delta_{X_2}\geq 2] 
        &\leq \left(1-\frac{2\rho+\alpha+2}{3\rho}\right)^2
        \leq \left(1-\frac{2\rho}{3\rho}\right)^2
        =\frac{1}{9}~,
    \end{align*}
    as required.
\end{proof}

\begin{theorem}
\label{thm:buffer-size-upper-bound}
For a subtree of size $n\leq \beta+\alpha$, we have
that the expected number of non-full blocks $\E[E_n]\leq 27 \cdot 
\delta_n$.
% \begin{align*}
%     \E[E_n] \leq
%     \begin{cases}
%       1 & n\leq \rho+\alpha \\
%       27 \cdot \delta_n & \text{o.w.}
%     \end{cases} ~.
% \end{align*}
\end{theorem}

\begin{proof}
    We will prove the statement by induction on $n$.
    Observation \ref{obs:buffer-delta-one-nonfull-blocks}, Lemma~\ref{lem:buffer-delta-two-nonfull-blocks}, and Lemma~\ref{lem:buffer-delta-three-nonfull-blocks} provide the bases of the induction.
    The problem is divided into subproblems by using random variables $X_1,\dots,X_{\delta_n}$.
    \begin{align*}
        \E[E_n] = \E[E_{X_1}] +\E[E_{X_2}] + \dots + \E[E_{X_{\delta_n}}] = \delta_n\E[E_{X_1}]
    \end{align*}
    Note that $X_i\leq n-\alpha$, so the induction assumption is applicable to it.
    For $n\geq 4\rho+\alpha$, we will use the induction assumption and bases to obtain
    \begin{align*}
        \E[E_{X_i}] &\leq \Pr[\delta_{X_i}=1]\cdot 1 + \Pr[\delta_{X_i}=2]\cdot 5 + \Pr[\delta_{X_i}=3]\cdot 10 + \sum_{k=4}^{\delta_n} \Pr[\delta_{X_i}=k]\cdot (27\cdot k)\\
        &< \Pr[\delta_{X_i}=1]\cdot 10 + \Pr[\delta_{X_i}=2] \cdot 10 + \Pr[\delta_{X_i}=3] \cdot 10 + \sum_{k=4}^{\delta_n} \Pr[\delta_{X_i}=k] \cdot (27 \cdot k)\\
        &= 10\Pr[1\leq\delta_{X_i}\leq 3] +27\sum_{k=4}^{\delta_n}\Pr[\delta_{X_i}=k] \cdot k\\
        &\leq 10+27\sum_{k=4}^{\delta_n}\Pr[\delta_{X_i}=k] \cdot k
    \end{align*}
    We can rewrite $\sum_{k=4}^{\delta_n} \Pr[\delta_{X_i}=k] \cdot k$ to $3\Pr[\delta_{X_i}\geq 4] + \sum_{k=4}^{\delta_n} \Pr[\delta_{X_i}\geq k]$.
    Then by using Corollary~\ref{simplified-exp-lem:upper-bound-on-accumulative-prob}, we get    
    \begin{align*}
        \E[E_{X_i}] &\leq 10 + 27(3\Pr[\delta_{X_i}\geq 4] + \sum_{k=4}^{\delta_n} \Pr[\delta_{X_i}\geq k])\\
        &\leq 10+27 \left( 3e^{-2}+\sum_{k=4}^{\infty}e^{-(k-2)} \right)\\
        &\leq 10+27 \left(\frac{3}{e^2}+\sum_{k=2}^{\infty}e^{-k} \right)\\
        &= 10+27 \left(\frac{3}{e^2}+\frac{e^{-2}}{1-e^{-1}} \right)
        \leq 10+0.6202\times 27 <27~,
    \end{align*}
    as stated.
\end{proof}

\section{Additional Proofs}
\label{sec:final-claim}
\begin{lemma} \label{lem:final-claim}
Consider a subtree with $n>\gamma$ keys, where $\gamma$ is equal to $\beta+\alpha=(\alpha+1)\rho+\alpha$.
We have:
\begin{align*}
2\frac{\alpha}{\eps}\binom{n-1}{\alpha-i} - 2\frac{\alpha}{\eps}a\binom{n-1-\gamma}{\alpha-i}
&\le \binom{n}{\alpha+1-i} - \binom{n-\gamma}{\alpha+1-i} - \gamma\binom{n-1-\gamma}{\alpha-i}~.
\end{align*}
\end{lemma}
\begin{proof}
    The proof is by induction on $n$. 
    For $n=\gamma+1$ the statement is
    \begin{align}
    2\frac{\alpha}{\eps}\binom{\gamma}{\alpha-i}
    &\le \binom{\gamma+1}{\alpha+1-i}\\
    \Longleftrightarrow\quad
    2\frac{\alpha}{\eps} \frac{\gamma(\gamma-1)\ldots(\gamma-\alpha+i+1)}{(\alpha-i)!}
    &\le\frac{(\gamma+1)\gamma\ldots(\gamma-\alpha+i+1)}{(\alpha-i+1)(\alpha-i)!}\\
    \Longleftrightarrow \quad
    2\frac{\alpha}{\eps}(\alpha-i+1) &\le \gamma + 1    \quad,
    \end{align}
    This is true because $\gamma=(\alpha+1)\rho+\alpha =108(\alpha+1)\alpha/\eps+\alpha>2(\alpha+1)\alpha/\eps$.

    Assume the lemma holds for each $\gamma<m \leq n$.
    We will prove the lemma for $n+1$, e.g. we will show: 
    \begin{align*}
    2\frac{\alpha}{\eps}\binom{n}{\alpha-i} - 2\frac{\alpha}{\eps}\binom{n+1-1-\gamma}{\alpha-i}
    &\le \binom{n+1}{\alpha+1-i} - \binom{n+1-\gamma}{\alpha+1-i} - \gamma\binom{n+1-1-\gamma}{\alpha-i}
    \end{align*}
    holds for all $i$.
    The inequality is true for $i=\alpha-1$, since
    \begin{align*}
    2\frac{\alpha}{\eps} n - 2\frac{\alpha}{\eps}(n-\gamma)
    &\le \frac{n(n+1)}{2} - \frac{(n-\gamma)(n-\gamma+1)}{2} - \gamma(n-\gamma)\\
    \Longleftrightarrow\quad
    2\frac{\alpha}{\eps} \gamma &\le \frac{2\gamma n - \gamma^2 +\gamma}{2} - \frac{2\gamma n -2\gamma^2}{2}\\
    \Longleftrightarrow\quad
    2\frac{\alpha}{\eps} \gamma &\le \frac{2n +\gamma^2 + \gamma}{2} \quad.
    \end{align*}
    The right-hand side is greater than $\frac{\gamma^2 + \gamma}{2}$.
    Hence, it suffices to observe that
    \begin{align*}
    2\frac{\alpha}{\eps} \gamma \le \frac{\gamma^2 + \gamma}{2}
    \quad
    \Longleftrightarrow\quad
    4\frac{\alpha}{\eps} \le \gamma+1 =(\alpha+1)\rho+\alpha+1 = (\alpha+1)108\alpha/\eps+\alpha+1
    \end{align*}
    which is true.

    For $i < \alpha-1$, we will use the identity equation $\binom{m}{k}=\binom{m-1}{k}+\binom{m-1}{k-1}$ together with the induction assumption.
    We have
    \begin{align*}
    &2\frac{\alpha}{\eps}\binom{n}{\alpha-i} - 2\frac{\alpha}{\eps}\binom{n-\gamma}{\alpha-i}\\
    &=2\frac{\alpha}{\eps} \binom{n-1}{\alpha-(i+1)} +2\frac{\alpha}{\eps} \binom{n-1}{\alpha-1} -2\frac{\alpha}{\eps}\binom{n-1-\gamma}{\alpha-(i+1)} -2\frac{\alpha}{\eps}\binom{n-1-\gamma}{\alpha-i}\\
    &= \left[2\frac{\alpha}{\eps}\binom{n-1}{\alpha-(i+1)} -2\frac{\alpha}{\eps}\binom{n-1-\gamma}{\alpha-(i+1)}\right] 
    +\left[2\frac{\alpha}{\eps}\binom{n-1}{\alpha-i} -2\frac{\alpha}{\eps}\binom{n-1-\gamma}{\alpha-i}\right]\\
    &\le
    \binom{n}{\alpha+1-(i+1)} - \binom{n-\gamma}{\alpha+1-(i+1)} - \gamma\binom{n-1-\gamma}{\alpha-(i+1)} \\&+ \binom{n}{\alpha+1-i} - \binom{n-\gamma}{\alpha+1-i} - \gamma\binom{n-1-\gamma}{\alpha-i}\\
    &= \binom{n+1}{\alpha+1-i} - \binom{n+1-\gamma}{\alpha+1-i} - \gamma\binom{n-\gamma}{\alpha-i} \quad.
    \end{align*}
    \end{proof}

\end{document}